\documentclass[reprint,
 amsmath,amssymb,
 aps,
]{revtex4-1}
\usepackage{amsmath,amsthm,amssymb,amsfonts,verbatim}

\usepackage{graphicx}
\usepackage{dcolumn}
\usepackage{bm}

\newtheorem{lemma}{Lemma}
\newtheorem{thm}[lemma]{Theorem}
\newtheorem{cor}[lemma]{Corollary}

\newtheorem{defn}{Definition}

\newtheorem{rmk}{Remark}
\newtheorem{exm}{Example}

\renewcommand{\epsilon}{\varepsilon}
\renewcommand{\le}{\leqslant}
\renewcommand{\ge}{\geqslant}

\def\CC{\mathbb{C}}

\def\ZZ{\mathbb{Z}}

\def\MM{\mathbb{M}}

\def \Tr {{\rm Tr}}
\def\cA{\bar{A}}

\def \bc {{\bf c}}
\def \bx {{\bf x}}

\def \bu {{\bf u}}
\def \bv {{\bf v}}
\def \bo {{\bf 0}}

\def \rGF {{\rm GF}}
\def\bp{{\bar{\phi}}}

\newcommand{\Ga}{\alpha}
\newcommand{\Gb}{\beta}

\newcommand{\Ge}{\epsilon}

\newcommand{\Gl}{\lambda}

\newcommand{\Gv}{\varphi}

\begin{document}

\preprint{APS/123-QED}

\title{Multipartite entangled states, symmetric matrices and error-correcting codes}

\author{Keqin Feng}
\affiliation{
 Department of Mathematical Sciences, Tsinghua University, Beijing, China
}%

\author{Lingfei Jin}%
 \email{lfjin@fudan.edu.cn}
\affiliation{%
Shanghai Key Laboratory of Intelligent Information Processing, School of Computer Science, Fudan University, Shanghai 200433, China
}%


\author{Chaoping Xing and Chen Yuan}
\affiliation{
 School of Physical \& Mathematical Sciences, Nanyang Technological University, Singapore
}%

\date{\today}

\begin{abstract}
A pure quantum state is called $k$-uniform if all its reductions to $k$-qudit are maximally mixed. We investigate the general constructions of $k$-uniform pure quantum states of $n$ subsystems with $d$ levels. We provide one construction via symmetric matrices and the second one through classical error-correcting codes. There are three main results arising from our constructions. Firstly, we show that for any given even $n\ge 2$,  there always exists an $n/2$-uniform $n$-qudit quantum state of level $p$ for sufficiently large prime $p$. Secondly, both constructions show that
their exist $k$-uniform $n$-qudit pure quantum states such that $k$ is proportional to $n$, i.e., $k=\Omega(n)$ although the construction from symmetric matrices outperforms the one by error-correcting codes. Thirdly, our symmetric matrix construction provides a positive answer to the open question in \cite{DA} on whether there exists $3$-uniform $n$-qudit pure quantum state for all $n\ge 8$. In fact, we can further prove that, for every  $k$, there exists a constant $M_k$ such that there exists a $k$-uniform $n$-qudit quantum state for all $n\ge M_k$. In addition, by using concatenation of algebraic geometry codes, we give an explicit construction of $k$-uniform quantum state when $k$ tends to infinity.
\end{abstract}

\pacs{Valid PACS appear here}
\maketitle


\section{Introduction}
Quantum entanglement appears in many areas of quantum information theory including quantum communications~\cite{Be1993,Be1996,Sch1995}, quantum computing~\cite{Joz1997,Joz2002,Vir2005} and quantum key distribution~\cite{Koa2004}. Quantum entanglement theory is developed to determine which states are entangled and which are separable. In bipartite entanglement, the simplest is quantum bipartite pure state. To determine whether this pure state is separable, we just diagonalize its reduced density matrix. But it is still NP-hard to determine separability in bipartite system~\cite{Gur}. In general problem of multipartite entanglement, besides separability and entanglement, there are many types of partial separability which complicates this problem. Although there are some attempts to detect genuine multipartite entanglement~\cite{Hub1,Hub2}, there are still many open problems in this area.

One of the intriguing problem is to investigate highly entangled states of several qubits\cite{Gis,Hig,Ken,Brow,Ost,Bor,Br,Mar,Tam}.
 In \cite{Fac1,Fac2}, they considered the one qubit reduced state which is maximally mixed. This idea was further developed by Arnaud and Cerf \cite{Arn}. They proposed the concept of $k$-multipartite maximally entangled pure states or $k$-uniform for short, i.e., any $k$-partite reduced state is maximally mixed.

 It was shown in \cite{Sc} that an $n$-qudit pure quantum state $|\Phi\rangle$ of level $d$ is $k$-uniform if and only if $|\Phi\rangle$ is a pure $((n,1,k+1))_d$ quantum error-correcting code. Using this connection the author was able to construct some $k$-uniform pure quantum states through stabilizer quantum codes obtained from classical self-dual codes. In \cite{DA}, a connection between $k$-uniform pure quantum states and orthogonal arrays was established and several classes of $k$-uniform states were constructed. More precisely speaking, the following result were obtained in \cite{DA}.
 \begin{itemize}
\item There exist $k$-uniform $(d+1)$-qudits states of  $d$ levels whenever $d\geq 2$ and $k\leq \frac{d+1}{2}$.
\item There exist $2$-uniform $n$-qudits states of  $2$ levels whenever  $n\ge 5$.
\item There exist $3$-uniform $(2^m+2)$-qudits states of  $2^m$ levels whenever  $m\ge 2$.
\item There exist $2^m-1$-uniform $(2^m+2)$-qudits states of  $2^m$ levels whenever  $m=2$ or $4$.
\end{itemize}
 In addition, some $k$-uniform $n$-qudits states of  $d$ levels were also given for some small values of $k,n,d$. The above special values of the parameters $k,n,d$ are obtained due to constraint from combinatorial structure of orthogonal arrays.

 In this paper, we first provide an equivalent definition for $k$-uniform quantum states through a map from $\ZZ_d^n$ to the complex numbers $\CC$. Based on this equivalent definition, we first derive a construction of $k$-uniform quantum states by using symmetric matrices. Again starting from this  equivalent definition, we present the second construction that makes use of classical error-correcting codes with good minimum distance and dual distance. There are three main results arising from our constructions. Firstly, we show that for any given even $n\ge 2$,  there always exists an $n/2$-uniform $n$-qudit quantum state of level $p$ for sufficiently large prime $p$. Secondly, both constructions show that
their exist $k$-uniform $n$-qudit pure quantum states such that $k$ is proportional to $n$, i.e., $k=\Omega(n)$ although the construction from symmetric matrices outperforms the one by error-correcting codes. Thirdly, our symmetric matrix construction provides a positive answer to the open question in \cite{DA} on whether there exists $3$-uniform $n$-qudit pure quantum states for all $n\ge 8$. In fact, we can further prove that, for every  $k$, there exists a constant $M_k$ such that there exists a $k$-uniform $n$-qudit quantum state for all $n\ge M_k$. In addition, by using concatenation of algebraic geometry codes, we give an explicit construction of $k$-uniform quantum state when $k$ tends to infinity. Both numeric and theoretic results reveal that the matrix construction is in general better than the one by classical error-correcting codes.

 The paper is organized as follows. In Section 2, we introduce basic definition of $k$-uniform quantum states and present an equivalent definition. By this  equivalent  definition, we propose two different constructions of $k$-uniform quantum states in Section 3. In Section 4, we investigate  the case where $n$ is small by presenting  some tables and a few other results. In the last section,  we discuss the case where $n$ tends to infinity, i.e., its asymptotic behavior through our construction. In addition, in this section we also provide an explicit construction of $k$-uniform quantum states based on our construction through error-correcting codes.

\section{Preliminaries on $k$-uniform quantum state}

\subsection{Definition}

A $k$-uniform $n$-qudit quantum state has the
property that, after tracing out all but $k$ qudits,  we are left
with the maximally mixed state for any $k$-tuple of qudits. This means that  all information about the system is lost after removal
of $n-k$ or more parties. Precisely speaking, a pure
quantum state of $n$ subsystems of  level $d$  is called $k$-uniform (or $k$-maximally entangled) if every reduction to
$k$ qudits is maximally mixed. Let us give a mathematical definition.

The density matrix of a quantum state  $|\Phi\rangle=\sum_{\bc\in\ZZ_d^n}\phi_{\bc}|\bc\rangle$ is defined by $\rho:=\sum_{\bc,\bc'\in\ZZ_d^n}\phi_{\bc}\bp_{\bc}|\bc\rangle\langle\bc'|$. For a subset $A$ of $\{1,2,\dots,n\}$ and a vector $\bc\in\ZZ_d^n$, we denote by $\bc_A$ the projection of $\bc$ at $A$.  The reduction of $|\Phi\rangle$ to $A$ has the density matrix $\rho_A:=\sum_{\bc,\bc'\in\ZZ_d^n}\phi_{\bc}\bp_{\bc}\langle\bc_{\cA}|\bc'_{\cA}\rangle|\bc_A\rangle\langle\bc'_A|$, where $\cA$ is the complement set of $A$ (i.e., $\cA=\{1,2,\dots,n\}\setminus A$) and $\langle\bc_{\cA}|\bc'_{\cA}\rangle$ is defined to be $1$ if $\bc_{\cA}=\bc'_{\cA}$ and $0$ otherwise.
\begin{defn}\label{defn:1} {\rm A pure quantum state  $|\Phi\rangle=\sum_{\bc\in\ZZ_d^n}\phi_{\bc}|\bc\rangle$ is called $k$-uniform if for any subset $A$ of $\{1,2,\dots,n\}$, the reduction of $|\Phi\rangle$ to $A$ has the density matrix $\rho_A=\Ga_A\sum_{\bc_A\in\ZZ_d^k}|\bc_A\rangle\langle\bc_A|$, where $\Ga_A\in\CC$ depends only on $A$ and  $|\Phi\rangle$.
}\end{defn}

\begin{exm}{\rm Consider $5$-qudit quantum state of level $2$
\begin{eqnarray*}|\Phi\rangle&=&-|00000\rangle+|01111\rangle-|10011\rangle+|11100\rangle\\&+&|00110\rangle+|01001\rangle+
|10101\rangle+|11010\rangle.
\end{eqnarray*}
Let $A=\{3,4\}$. Then an easy computation shows that the density matrix $\rho_A$ is $2|00\rangle\langle00|+2|01\rangle\langle01|+2|10\rangle\langle10|+2|11\rangle\langle11|$. One can also verify that the density matrix $\rho_A$ has the same form for all other subsets $A$ with $|A|=2$. By definition, $|\Phi\rangle$ is $2$-uniform.
}\end{exm}

The well-known Greenberger-Horne-Zeilinger  states belong
to the class $1$-uniform, while W states do not belong to any class of
k-uniform states. For a state  $|\Phi\rangle$,   a multipartite entanglement measures $Q_k( |\Phi\rangle)$ was defined \cite{Sc}. The original
Meyer-Wallach measure $Q_1( |\Phi\rangle)$ is actually the average entanglement between
individual qudits and the rest. As $k$
increases,  $Q_k( |\Phi\rangle)$ is getting more sensitive to correlations
of an increasingly global nature. $Q_k( |\Phi\rangle)$ is upper bounded by $1$. It was proved in \cite[Propsoition 2]{Sc} that  $|\Phi\rangle$ is $k$-uniform if and only if $Q_k( |\Phi\rangle)=1$.
\subsection{An equivalent definition}
An $n$-qudit quantum state
$|\Phi\rangle=\sum_{\bc\in\ZZ_d^n}\phi_{\bc}|\bc\rangle$ of level $d$  is associated with a map $\varphi$ from $\ZZ_d^n$ to $\CC$ given by $\varphi(\bc)=\phi_{\bc}$. This means that  $n$-qudit quantum states of level $d$ are identified with maps from $\ZZ_d^n$ to $\CC$. Thus,  an $n$-qudit state $|\Phi\rangle$ can be written as $\sum_{\bc\in\ZZ_d^n}\varphi(\bc)|\bc\rangle$ for a given function $\varphi$.  A $k$-uniform  quantum state can be described in terms of its associated map $\varphi$.
\begin{lemma}\label{lm:1}
An $n$-qudit state $|\Phi\rangle=\sum_{\bc\in\ZZ_d^n}\varphi(\bc)|\bc\rangle$ is $k$-uniform if and only if
\begin{itemize}
  \item [(i)] $\varphi$ is not identical to zero.
  \item [(ii)] For any subset $A$ of $\{1,2,\cdots,n\}$ with  $|A|=k$, and every $c_A,c'_{A}\in \ZZ_d^k$, one has
      \[\sum_{\bc_{\cA}\in\ZZ_d^{n-k}}\overline{\varphi(c_A,c_{\cA})}\varphi(c'_A,c_{\cA})=\left\{
        \begin{array}{ll}
         0, & \mbox{if $c_A\neq c'_A$,} \\
          \frac{\langle \Phi|\Phi\rangle}{d^k}, & \mbox{if $c_A=c'_A$.}
        \end{array}
      \right.\]
\end{itemize}
\end{lemma}
\begin{proof}{\rm
If $|\Phi\rangle$ is $k$-uniform, by tracing out any $n-k$ qudits, the $k$-qudit reduced density matrix is proportional to identity matrix.
We fix a subset $A$ of $\{1,2,\cdots,n\}$ with  $|A|=k$.  An $n$-qudit state $|\Phi\rangle=\sum_{\bc\in\ZZ_d^n}\varphi(\bc)|\bc\rangle$ is written as
\begin{equation*}
|\Phi\rangle=\sum_{\bc_A\in\ZZ_d^k,\bc_{\cA}\in \ZZ_d^{n-k}}\varphi(\bc_A,\bc_{\cA})|\bc_A\rangle\langle\bc_{\cA}|.
\end{equation*}
Denote by $\rho$ the density matrix of $|\Phi\rangle$, \emph{i.e.,} $\rho=|\Phi\rangle\langle\Phi|$.
Consider the reduced state
\begin{eqnarray*}
\rho_A&=&{\rm Tr}_{\cA}(\rho)\\&=&\sum_{\bc_A,\bc_A'\in \ZZ_d^{k}}|\bc_A\rangle\langle\bc_A'|\sum_{\bc_{\cA},\bc_{\cA}'\in \ZZ_d^{k}}\overline{\varphi(c_A,c_{\cA})}\varphi(c'_A,c_{\cA}')\langle \bc_{\cA}|\bc_{\cA}'\rangle \\
&=&\sum_{\bc_A,\bc_A'\in \ZZ_d^{k}}|\bc_A\rangle\langle\bc_A'|\sum_{\bc_{\cA}\in \ZZ_d^{k}}\overline{\varphi(c_A,c_{\cA})}\varphi(c'_A,c_{\cA})
\end{eqnarray*}
Since $|\Phi\rangle$ is $k$-uniform, the reduced state $\rho_A$ is proportional to identity matrix. The sum of diagonal element of $\rho_A$ is $\langle \Phi|\Phi\rangle$ which implies that
\[\sum_{\bc_{\cA}\in\ZZ_d^{n-k}}\overline{\varphi(c_A,c_{\cA})}\varphi(c'_A,c_{\cA})=\left\{
        \begin{array}{ll}
         0, & \mbox{if $c_A\neq c'_A$,} \\
          \frac{\langle \Phi|\Phi\rangle}{d^k}, & \mbox{if $c_A=c'_A$.}
        \end{array}
\right.\]
Vice versa, we have the desired result.
}\end{proof}

\section{Constructions of $k$-uniform quantum states}
Before starting our first construction, we prove a lemma.
\begin{lemma}\label{lm:b2} Let $d\ge 2$  be an integer. Let $a_1,a_2,\dots,a_m$ be $m$ integers. Assume that $\gcd(a_1,a_2,\dots,a_m,d)=e<d$. Then for every $b\in\ZZ_{d/e}$ the equation $a_1x_1+a_2x_2+\dots+a_mx_m\equiv be\bmod{d}$ has exactly $ed^{m-1}$ solutions in $\ZZ_d^m$.
\end{lemma}
\begin{proof}{\rm For $b\in\ZZ_{d/e}$, we denote by $N_b$ the number of solutions $\bx=(x_1,x_2,\dots,x_m)\in\ZZ_d^m$ of   $a_1x_1+a_2x_2+\dots+a_mx_m\equiv be\bmod{d}$.  We claim that $N_b=N_0$ for any $b\in\ZZ_{d/e}$.  Denote by $g$ the greatest common divisor $\gcd(a_1,a_2,\dots,a_m)$. Then one can find integers $u_1,u_2,\dots,u_m$ such that $a_1u_1+a_2u_2+\dots+a_mu_m=g$.
Since $\gcd(g,d)=e$, we can find $c$ such that $cg\equiv e\bmod{d}$. Thus,  $a_1(cu_1)+a_2(cu_2)+\dots+a_m(cu_m)=cg\equiv e\bmod{d}$. If $\bu\in\ZZ_d^m$ is a solution of   $a_1x_1+a_2x_2+\dots+a_mx_m\equiv 0\bmod{d}$, then $\bv+(bcu_1,bcu_2,\dots,bcu_m)$ is a solution of  $a_1x_1+a_2x_2+\dots+a_mx_m\equiv b\bmod{d}$. This implies that $N_0\le N_b$. On the other hand, if $\bv\in\ZZ_d^m$ is a solution of   $a_1x_1+a_2x_2+\dots+a_mx_m\equiv b\bmod{d}$, then $\bv-(bcu_1,bcu_2,\dots,bcu_m)$ is a solution of  $a_1x_1+a_2x_2+\dots+a_mx_m\equiv 0\bmod{d}$. This implies that $N_b\le N_0$. Now we have $d^m=\sum_{b\in\ZZ_d}N_d=\frac de\times N_0$. The desired result follows.
}\end{proof}

 Based on Lemma \ref{lm:1}, we first provide a construction of $k$-uniform quantum state through symmetric matrices. Our map $\varphi$ is in fact a quadratic function.
Let $\zeta_d$ denote a $d$th primitive root of unity in $\CC$. For two subsets $A,B\subseteq \{1,2,\dots,n\}$ and a matrix $H=(h_{ij})\in \MM_{n\times n}(\ZZ_d)$, we denote by  $H_{A\times B}$ the submatrix $(h_{ij})_{i\in A,j\in B}$. An $n\times n$ matrix $E$ over $\ZZ_d$ is called invertible if there exists an $n\times n$ matrix $G$ over $\ZZ_d$ such that $EG=GE$ is equal to the identity matrix. It is well known that $E$  is  invertible if and only if the determinate of $E$ is co-prime with $d$. If $E$  is  invertible, then for any nonzero vector $\bc\in\ZZ_d^n$, we must have $\bc E\neq\bo$. Otherwise, one would have $\bo=\bo E^{-1}=\bc EE^{-1}=\bc$.
\begin{thm}\label{thm:a1}
If there is a zero diagonal symmetric matrix $H\in\MM_{n\times n}(\ZZ_d)$ such that for any subset $A$ of $\{1,2,\dots,n\}$ with $|A|=k$, there exists a subset $B$ of $\cA$  with $|B|=k$ such that the submatrix $H_{A\times B}$ is a $k\times k$ invertible matrix over $\ZZ_d$ , then the  $n$-qudit state $|\Phi\rangle=\sum_{\bc\in\ZZ_d^n}\varphi(\bc)|\bc\rangle$ is $k$-uniform with $\varphi(\bc)=\zeta_d^{\bc \widetilde{H}\bc^T}$, where $\widetilde{H}=(\widetilde{h}_{ij})$ with $\widetilde{h}_{ij}=h_{ij}$ for $i<j$ and $0$ otherwise.
\end{thm}
\begin{proof}{\rm Consider the map $f$ from $\ZZ_d^n$ to $\ZZ_d^n$ given by $f(\bc)= {\bc \widetilde{H}\bc^T}$. Then for every subset $A$ of  $\{1,2,\cdots,n\}$ with $|A|=k$, we have
\begin{eqnarray*}
f(\bc_A,\bc_{\cA})&=&\bc_A(\widetilde{H}_{A\times {\cA}}+\widetilde{H}_{{\cA}\times A}^{T})\bc_{\cA}^T+\bc_A\widetilde{H}_{A\times A}\bc_A^T+\bc_{\cA}\widetilde{H}_{{\cA}\times {\cA}}\bc_{\cA}^T \\ &=&\bc_A H_{A\times {\cA}}\bc_{\cA}^T+\bc_A\widetilde{H}_{A\times A}\bc_A^T+\bc_{\cA}\widetilde{H}_{{\cA}\times {\cA}}\bc_{\cA}^T.
\end{eqnarray*}
Hence,
\begin{eqnarray*}\label{eq:2.1}&&f(\bc_A,\bc_{\cA})-f(\bc'_A,\bc_{\cA})\\&=&(\bc_A-\bc'_A){H}_{A\times {\cA}}\bc_{\cA}^T+\bc_A\widetilde{H}_{A\times A}\bc_A^T-\bc'_A\widetilde{H}_{A\times A}(\bc'_A)^T.\end{eqnarray*}
If $\bc_A=\bc'_A$, one has
\begin{eqnarray*}\sum_{\bc_{\cA}\in\ZZ_d^{n-k}}\overline{\varphi(c_A,c_{\cA})}\varphi(c'_A,c_{\cA})&=&\sum_{\bc_{\cA}\in\ZZ_d^{n-k}}\zeta_d^{f(\bc_A,\bc_{\cA})}
\zeta_d^{-f(\bc_A,\bc_{\cA})}\\&=&d^{n-k}= \frac{\langle \Phi|\Phi\rangle}{d^k}.\end{eqnarray*}
Note that $\langle \Phi|\Phi\rangle=d^n$.

If $\bc_A\neq\bc'_A$, then $(\bc_A-\bc'_A){H}_{A\times B}$ is not the zero vector and hence $(\bc_A-\bc'_A){H}_{A\times {\cA}}$ (denoted  by $(a_1,a_2,\dots,a_{n-k})$) is a nonzero vector in $\ZZ_d^{n-k}$. Let $e$ be $\gcd(a_1,a_2,\dots,n-k,d)$. Then $e<d$.  By Lemma \ref{lm:b2}, $(\bc_A-\bc'_A){H}_{A\times {\cA}}\bx=be$ has exactly $ed^{n-k-1}$ solutions in $\ZZ_d^{n-k}$ for every $b\in\ZZ_{d/e}$. Hence, by \eqref{eq:2.1}, we have
\begin{eqnarray*}\sum_{\bc_{\cA}\in\ZZ_d^{n-k}}\overline{\varphi(c_A,c_{\cA})}\varphi(c'_A,c_{\cA})&=&\zeta_d^g\sum_{\bc_{\cA}\in\ZZ_d^{n-k}}\zeta_d^{(\bc_A-\bc'_A){H}_{A\times {\cA}}\bc_{\cA}}\\
&=&ed^{n-k-1}\zeta_d^g\sum_{b=0}^{d/e-1}\zeta_d^{be}=0,\end{eqnarray*}
where $g=\bc_A\widetilde{H}_{A\times A}\bc_A^T-\bc'_A\widetilde{H}_{A\times A}(\bc'_A)^T$. This completes the proof.
}\end{proof}

If $d$ is a prime $p$, then the condition in Theorem \ref{thm:a1} can be simplified.
\begin{thm}\label{thm:1} Let $p$ be a prime.
If there is a zero diagonal symmetric matrix $H\in\MM_{n\times n}(\ZZ_p)$ such that for any subset $A$ of $\{1,2,\dots,n\}$ with $|A|=k$,  the submatrix $H_{A\times \cA}$ has rank $k$, then the  $n$-qudit state $|\Phi\rangle=\sum_{\bc\in\ZZ_p^n}\varphi(\bc)|\bc\rangle$ is $k$-uniform with $\varphi(\bc)=\zeta_d^{\bc \widetilde{H}\bc^T}$, where $\widetilde{H}=(\widetilde{h}_{ij})$ with $\widetilde{h}_{ij}=h_{ij}$ for $i<j$ and $0$ otherwise.
\end{thm}
\begin{proof}{\rm In this case, $H_{A\times \cA}$ has an invertible submatrix $H_{A\times B}$ for some subset $B$ of $\cA$ with $|B|=k$. The desired result follows from Theorem \ref{thm:a1}.
}\end{proof}

\begin{exm}{\rm
Based on Thorem \ref{thm:a1}, we provide two examples for $1$-uniform $2$-qudit quantum states,  with level $4$ and the other with level $6$. In both cases, the matrix is given by
\begin{equation*}
H=\left(
                                                                               \begin{array}{cc}
                                                                                 0 & 1 \\
                                                                                 1 & 0 \\
                                                                               \end{array}
                                                                             \right)
\end{equation*}
The quantum state of level $4$ corresponding to this matrix is
\begin{eqnarray*}
&&|\Phi\rangle=|00\rangle +|10\rangle +|20\rangle +|30\rangle +|01\rangle +i|11\rangle -|21\rangle -i|31\rangle \\&&+|02\rangle -|12\rangle +|22\rangle -|32\rangle +|03\rangle -i|13\rangle -|23\rangle +i|33\rangle \\&&
\end{eqnarray*}
The quantum state of level $6$ corresponding to this matrix is
\begin{eqnarray*}
&&|\Phi\rangle=|00\rangle +|10\rangle +|20\rangle +|30\rangle +|40\rangle +|50\rangle +|01\rangle +\zeta_6|11\rangle \\&&+\zeta_6^2|21\rangle +\zeta_6^3|31\rangle +\zeta_6^4|41\rangle +\zeta_6^5|51\rangle +|02\rangle +\zeta_6^2|12\rangle +\zeta_6^4|22\rangle \\&& +|32\rangle+\zeta_6^2|42\rangle +\zeta_6^4|52\rangle +|03\rangle +\zeta_6^3|13\rangle +|23\rangle +\zeta_6^3|33\rangle \\&&+|43\rangle +\zeta_6^3|53\rangle +|04\rangle +\zeta_6^4|14\rangle +\zeta_6^2|24\rangle +|34\rangle +\zeta_6^4|44\rangle \\&&+\zeta_6^2|54\rangle +|05\rangle +\zeta_6^5|15\rangle +\zeta_6^4|25\rangle +\zeta_6^3|35\rangle +\zeta_6^2|45\rangle +\zeta_6|55\rangle
\end{eqnarray*}
}\end{exm}
The second construction applies Lemma \ref{lm:1} to linear codes with good minimum distance and dual distance. As our classical codes are defined over prime fields $\ZZ_p$, we consider level $p$ only for primes $p$ for the following constrution.
\begin{thm}\label{thm:2}
If $C$ is a $p$-ary linear code of length $n$. Let $d$ and $d^{\perp}$ be the minimum distance of $C$ and  its Euclidean dual $C^{\perp}$, respectively. If $\min\{d,d^{\perp}\}\ge k+1$, then $|\Phi\rangle=\frac{1}{\sqrt{|C|}}\sum_{\bc\in C}|\bc\rangle$ is $k$-uniform $n$-qudit quantum state of level $p$.
\end{thm}
\begin{proof} {\rm It is clear that $\langle \Phi|\Phi\rangle$ is equal to $1$. Define the map  $\Gv$ from $\ZZ_p^n$ to $\CC$ given by $\Gv(\bx)= 1/\sqrt{|C|}$ if $\bx\in C$ and $0$ otherwise. Consider a subset $A$ of $\{1,2,\cdots,n\}$ with $|A|=k$.

Since $d^{\perp}\ge k+1$, for every $\bc_A\in\ZZ_p^k$ there are exactly $|C|/p^{k}$ vectors $\bc_{\cA}\in\ZZ_p^{n-k}$ such that $(\bc_A,\bc_{\cA})\in C$. Thus,
If $\bc_A=\bc'_A$, one has
\begin{eqnarray*}&&\sum_{\bc_{\cA}\in\ZZ_p^{n-k}}\overline{\varphi(c_A,c_{\cA})}\varphi(c'_A,c_{\cA})=\sum_{(\bc_A,\bc_{\cA})\in C}\frac1{{|C|}}\\ &&=\frac{|C|/p^{k}}{|C|}= \frac{\langle \Phi|\Phi\rangle}{p^k}.\end{eqnarray*}
If $\bc_A\neq\bc'_A$, then the Hamming distance between $(\bc_A,\bc_{\cA})$ and $(\bc'_A,\bc_{\cA})$ is at most $k$. This implies that $(\bc_A,\bc_{\cA})$ and $(\bc'_A,\bc_{\cA})$ do not belong to $C$ simultaneously for any $\bc_{\cA}\in\ZZ^{n-k}_p$. In other words, $\overline{\varphi(c_A,c_{\cA})}\varphi(c'_A,c_{\cA})=0$ for any $\bc_{\cA}\in\ZZ^{n-k}_p$. In this case, we have $\sum_{\bc_{\cA}\in\ZZ_p^{n-k}}\overline{\varphi(c_A,c_{\cA})}\varphi(c'_A,c_{\cA})=0$.
The desired result follows from Lemma \ref{lm:1}.
}\end{proof}

\begin{rmk} {\rm \begin{itemize}
\item[(i)] In general, the construction in Theorems \ref{thm:a1} and  Theorems \ref{thm:1}  gives better results than the one in Theorem \ref{thm:2}. We will see this in Sections 3 and 4.
 \item[(ii)] For the construction in Theorem \ref{thm:2}, we require linear codes with both good minimum distance and dual distance. Algebraic geometry codes provide an excellent family of codes with good minimum distance and dual distance \cite{Stich}. We will illustrate examples by algebraic geometry codes later in this section and the next two sections.
\end{itemize}
}\end{rmk}
\begin{cor}\label{cor:1} If there exists a $p$-ary $[n,n/2,\ge k+1]$-self-dual code, then
\begin{itemize}
\item[{\rm (i)}] there exists a $k$-uniform $n$-qudit quantum state of level $p$;
\item[{\rm (ii)}] there exists a $(k-1)$-uniform $(n-1)$-qudit quantum state of level $p$.
\end{itemize}
\end{cor}
\begin{proof}{\rm  Part (i) follows from Theorem \ref{thm:2} immediately.

For Part (ii), let $C$ be a $p$-ary $[n,n/2,\ge k+1]$-self-dual code. Without loss of generality, we may assume that there is a codeword $\bc$ of $C$ such that the last coordinate is not zero. Let $C_1$ consist of all codewords of $C$ whose last coordinates are zero. Then $C_1$ is $p$-ary linear code of dimension $n/2-1$, length $n$ and minimum distance at least $k+1$. Delete the last coordinate of $C_1$ to obtain a $p$-ary $[n-1,n/2-1,\ge k+1]$-linear code $C_2$. It is clear that the dual code $C_2^{\perp}$ is the code obtained from $C^{\perp}$ by deleting the last coordinate. It is clear that $C_2^{\perp}$ is a $p$-ary $[n-1,n/2,\ge k]$-linear code. Applying Theorem \ref{thm:2} to $C_2$ gives the desired result of Part (ii).
}\end{proof}


Theorem \ref{thm:2}  provides an explicit construction of $k$-uniform quantum states . We give an example below.
\begin{exm}{\rm
Consider the binary $[8,4,4]$-self-dual code C with generator matrix
\[\left(
  \begin{array}{cccccccc}
    1 & 0 & 0 & 0 & 1 & 0 & 1 & 0 \\
    0 & 1 & 0 & 0 & 0 & 1 & 1 & 0 \\
    0 & 0& 1 & 0 & 0 & 0 & 1 &1 \\
    0 & 0 & 0 & 1 & 1 & 1 & 1 & 1 \\
  \end{array}
\right)\]
Then the 8-qudit state $|v\rangle=\frac{1}{4}(|00000000\rangle+|11111111\rangle+|110000111\rangle+|01001011\rangle+|00101101\rangle+|00011110\rangle+|11001100\rangle+|10101010\rangle+|10011001\rangle+
|01100110\rangle+|01010101\rangle+|00110011\rangle+|11100010\rangle+|01111000\rangle+|11010001\rangle+|10110100\rangle)$
is 3-uniform.
}\end{exm}

\begin{rmk}{\rm \begin{itemize}
\item[(i)] Self-dual codes have been used to construct $k$-uniform quantum states in \cite{Sc,DA}. However, as we remarked,  the construction in Theorem \ref{thm:1}   gives better results than the one in Theorem \ref{thm:2}. Consequently,   the construction in Theorem \ref{thm:1}   gives better results than those from self-dual codes.
  \item[(ii)]  In fact, Theorem \ref{thm:2} does not require  codes to be self-dual. We now give an example showing that  Theorem \ref{thm:2} can give a better $k$-uniform quantum states than those from self-dual codes. We illustrate this by algebraic geometry codes in the following example.
\end{itemize}
}\end{rmk}
\begin{exm} \label{exm:1} {\rm We refer to \cite{Stich} for background on algebraic curves over finite fields and algebraic geometry codes. It is well know that an algebraic curve over the Galois field $\rGF(q)$ of $q$ elements with $n$ rational points and genus $g$ gives a $q$-ary linear code $C$ with parameters $[n,k,n-k+1-g]$ and its dual $C^{\perp}$ with parameters $[n,n-k,k+1-g]$ for any $g\le k\le n$.
\begin{itemize}
\item[(i)] By the online table \cite{Geer}, there is an algebraic curve over $\ZZ_5$ with $10$ rational points and genus $1$. Thus, one obtains a $5$-ary $[10,5,5]$ code $C$ and its dual code $C^{\perp}$ also has parameters $[10,5,5]$. By Theorem \ref{thm:2}, one obtains a $4$-uniform $10$-qudit quantum state of level $5$. On the other hand, the optimal $5$-ary self-dual code of length $10$ has minimum distance $4$ (see the online table \cite{GO}). Thus, applying Corollary \ref{cor:1} gives only  a $3$-uniform $10$-qudit quantum state of level $5$.
\item[(ii)] The above example is not a singularity. We can find other examples showing that  Theorem \ref{thm:2} can give better result than  Corollary \ref{cor:1}. Here is another example.  By the online table \cite{Geer}, there is an algebraic curve over $\ZZ_7$ with $16$ rational points and genus $2$. Thus, one obtains a $5$-ary $[16,8,7]$ code $C$ and its dual code $C^{\perp}$ also has parameters $[16,8,7]$. By Theorem \ref{thm:2}, we get a $6$-uniform $16$-qudit quantum state of level $7$. On the other hand, the optimal $7$-ary self-dual code of length $16$ has minimum distance $6$ (see the online table \cite{GO}). Thus, applying Corollary \ref{cor:1} gives only  a $5$-uniform $16$-qudit quantum state of level $7$.
    \end{itemize}
}\end{exm}

\section{The case where $n$ is small}
For given $d$ and $n$, one natural question is what is the maximal $k$ such that there exists a $k$-uniform $n$-qudit quantum state of level $d$. This question motivates the following definition.
\begin{defn}{\rm
For given positive integers $n\ge 2$ and $d\ge 2$, define $k_d(n)$ to be the largest $k$ such that there is an $n$-qudit state of level $d$ that is $k$-uniform.
}\end{defn}

One obvious upper bound on $k_d(n)$ is $n/2$. In this section, we will study some lower bounds on $k_d(n)$ by constructing $k$-uniform $n$-qudit quantum states of level $d$ via our results in Section 3. We discuss the cases for small $d$  and large $d$  separately. Although our matrix construction works well for composite levels $d$,  for simplicity  we only consider the case where $d=p$ is a prime number.

By Theorem \ref{thm:1}, in order to construct a $k$-uniform quantum state,  it is sufficient to find an $n\times n$ zero-diagonal  matrix  $H$  satisfying that $H_{A\times {\cA}}$ has rank $k$   for any subset $A$ of $\{1,2,\dots,n\}$ with $|A|=k$. Through the random matrix counting, we provide a sufficient condition for existence of such a matrix.

\begin{lemma}\label{lm:2} The number of $n\times n$ zero-diagonal  matrices  $H$ over $\ZZ_p$ satisfying that $H_{A\times {\cA}}$ has rank $k$   for any  subset $A$ of $\{1,2,\dots,n\}$ with $|A|=k$ is at least
\begin{equation}\label{eq:a2}p^{{n\choose n/2}}\left(1-{n\choose k}\left(1-\prod_{i=0}^{k-1}\left(1-\frac1{p^{n-k-i}}\right)\right)\right).\end{equation}
\end{lemma}
\begin{proof}{\rm
Consider the set ${\mathcal S}$ of $n\times n$ zero diagonal symmetric matrices over $\ZZ_p$. Then the cardinality of ${\mathcal S}$ is $p^{{n\choose k}}$. For a given  subset $A$ of $\{1,2,\dots,n\}$ with $|A|=k$, the set
$\{H\in {\mathcal S}:\; H_{A\times {\cA}}\ \mbox{ is invertible}\}$ has size $\prod_{i=0}^{k-1}(p^{n-k}-p^i) \times p^{{n\choose n/2}-k(n-k)}$. This implies that the set
 $\{H\in {\mathcal S}:\; H_{A\times {\cA}}\ \mbox{ is not invertible}\}$  has size $p^{{n\choose n/2}}- p^{{n\choose n/2}-k(n-k)}\times \prod_{i=0}^{k-1}(p^{n-k}-p^i) $. By the union bound, the number of zero diagonal symmetric matrices $H$  satisfying that,   for any  subset $A$ of  $\{1,2,\dots,n\}$ with $|A|=k$,  $H_{A\times {\cA}}$ \mbox{ is invertible} is at least $p^{{n\choose n/2}}-{n\choose k}\left( p^{{n\choose n/2}}- p^{{n\choose n/2}-k(n-k)}\times\prod_{i=0}^{k-1}(p^{n-k}-p^i) \right)$.
The desired result follows.
}\end{proof}

\begin{cor}\label{cor:3} If the triple $(n,k,p)$ satisfies ${n\choose k}(p^k-1)\le (p-1)p^{n-k}$, then there exists a $k$-uniform $n$-qudit quantum state of level $p$.
\end{cor}
\begin{proof}{\rm Denote by $N_p(n,k)$ the number in  \eqref{eq:a2}.
 By Theorem \ref{thm:1} and Lemma \ref{lm:2}, it is sufficient to show that $N_p(n,k)>0 $ under the condition of this Corollary. Indeed
\begin{eqnarray*}N_p(n,k)&>& p^{{n\choose n/2}}\left(1-{n\choose k}\sum_{i=0}^{k-1}\frac1{p^{n-k-i}}\right)\\&=& p^{{n\choose n/2}}\left(1-{n\choose k}\times \frac1{p^{n-k}}\times\frac{p^k-1}{p-1}\right)\ge 0.\end{eqnarray*}
This completes the proof.
}\end{proof}

\subsection{Small $k$}

In \cite{Sc}, it was proved that for any $n\ge5$, there exists a $2$-uniform $n$-qudit quantum state of level $2$. By using the above corollary, we can extend this results largely. For instance, we have the following result.

\begin{thm} For any prime $p$, there exists an integer $M_k$ such that for any $n\ge M_k$, one can construct a $k$-uniform $n$-qudit of level $p$. Furthermore, one has the following quantum states.
\begin{itemize}
\item[{(i)}] For any $n\ge 8$ and integer $k\ge 1$, there exists a  $3$-uniform $n$-qudit quantum state of level $p$.
\item[{(ii)}] For any $n\ge 12$, there exists a  $4$-uniform $n$-qudit quantum state of level $p$.
\item[{(iii)}] For any $n\ge 18$, there exists a  $5$-uniform $n$-qudit quantum state of level $p$.
\end{itemize}
\end{thm}
\begin{proof}{\rm Recall that $N_p(n,k)$ denotes the number in  \eqref{eq:a2}. We also note that for fixed $n$ and $k$, $N_p(n,k)$ monotonically decreases when $p$ increases. By Corollary \ref{cor:3}, for a fixed $k$,  $N_2(n,k)>0$ for all sufficiently large $n$, i.e., there exists an integer $M_k$ such that $N_2(n,k)>0$ for  any $n\ge M_k$. Hence, $N_p(n,k)>0$ for  any $p\ge 2$ and $n\ge M_k$. This completes the proof for the first part.

A simple calculation shows that $N_2(n,3)>0$ for all $n\ge 12$, $N_p(n,3)>0$ for all $n\ge 8$. By computer search, we find that $k_n(2)\ge 3$ for $8\le n\le 11$ (see Table I below). This completes the proof of Part (i).

The similar arguments apply to the proof of Parts (ii) and (iii).
}\end{proof}
\subsection{$p=2,3,5,7$}

Through computer search, we are able to find some lower bounds on $k_d(n)$. Due to our computation limitation, $n$ is limited to $24$ or smaller. Table I provides  lower bounds on $k_p(n)$ via construction of Theorem \ref{thm:1}. The entries with ``-" in Table I means that our computation limit does not allow us to find a reasonable lower bound on the corresponding $k_p(n)$.

\begin{center}
Table I\\ \medskip

\begin{tabular}{||c|c|c|c|c|c|c|c|c|c|c|c|c|c|c|c|c||} \hline
$n$ &$2$ &$3$ & $4$ & $5$ & $6$ & $7$ &$8$ &$9$ & $10$ & $11$ & $12$ & $13$  &$14$ &$15$ & $16$ & $17$  \\ \hline
$k_2(n)$ &$1$ &$1$ & $1$ & $2$ & $3$ & $2$ &$3$ &$3$ & $3$ & $3$ & $4$ & $4$  &$4$ &$4$ & $4$ & $4$  \\ \hline
$k_3(n)$ &$1$ &$1$ & $2$ & $2$ & $3$ & $3$ &$3$ &$3$ & $4$ & $4$ & $4$ & $4$  &$5$ &$5$ & $5$ & $5$ \\ \hline
$k_5(n)$ &$1$ &$1$ & $2$ & $2$ & $3$ & $3$ &$3$ &$4$ & $4$ & $4$ & $5$ & $5$  &$5$ &$5$ & $6$ & $6$  \\ \hline
$k_7(n)$ &$1$ &$1$ & $2$ & $2$ & $3$ & $3$ &$3$ &$4$ & $4$ & $4$ & $5$ & $5$  &$5$ &$6$ & $6$ & $7$  \\ \hline\hline
$n$& $18$ & $19$   &$10$ &$21$ & $22$ & $23$ & $24$&&&&&&&&& \\ \hline
$k_2(n)$ & $5$& $5$   &$5$ &$5$ & $5$ & $5$ & $6$ &&&&&&&&& \\ \hline
$k_3(n)$  & $5$ & $5$   &- &- & - & - & -  &&&&&&&&&\\ \hline
$k_5(n)$ & $6$ & - &- &- & - & - & - &&&&&&&&& \\ \hline
$k_7(n)$& $7$  & - &- &- & - & - & - &&&&&&&&&  \\ \hline
\end{tabular}
\end{center}

Table I provides lower bounds $k_p(n)$ only for prime levels $p$. As our Theorem \ref{thm:a1} works well for composite levels $d$, we give another table showing  lower bounds $k_d(n)$ for $d=4,6,8$ and $9$. Due to our computation limitation, we compute $k_d(n)$ only up to $n=10$.

\begin{center}
Table II\\ \medskip
\begin{tabular}{||c|c|c|c|c|c|c|c|c|c||}
  \hline
  $n$ & $2$ & $3$ & $4$ & $5$ & $6$ & $7$ & $8$ & $9$ & $10$ \\\hline
  $k_4(n)$ & $1$ & $1$ & $1$ & $2$ & $3$ & $2$ & $3$ & $3$ & $3$ \\\hline
  $k_6(n)$ & $1$ & $1$ & $1$ & $2$ & $3$ & $2$ & $3$ & $3$ & $3$ \\\hline
  $k_8(n)$ & $1$ & $1$ & $1$ & $2$ & $3$ & $2$ & $3$ & $3$ & $3$ \\\hline
  $k_9(n)$ & $1$ & $1$ & $2$ & $2$ & $3$ & $3$ & $3$ & $3$ & $4$ \\\hline
\end{tabular}
\end{center}

In addition, we provide one matrix that gives a 3-uniform 6-qudit quantum states of level 2. As our construction in Theorem \ref{thm:1} is explicit, the quantum state can be explicitly written down as long as the corresponding matrix is known.

\[
\left(
  \begin{array}{cccccc}
    0 & 1 & 1 & 1 & 0 & 0 \\
    1 & 0 & 0 & 1 & 0 & 1 \\
    1 & 0 & 0 & 1 & 1 & 0 \\
    1 & 1 & 1 & 0 & 1 & 1 \\
    0 & 0 & 1 & 1 & 0 & 1 \\
    0 & 1 & 0 & 1 & 1 & 0 \\
  \end{array}
\right)
\]

The corresponding quantum state for the above matrix is
\begin{eqnarray*}
&&|\Phi\rangle=|000000\rangle +|100000\rangle +|010000\rangle -|110000\rangle  +|001000\rangle \\&&
-|101000\rangle +|011000\rangle +|111000\rangle +|000100\rangle -|100100\rangle \\&&-|010100\rangle
-|110100\rangle -|001100\rangle -|101100\rangle +|011100\rangle \\&&-|111100\rangle+|000010\rangle
 +|100010\rangle +|010010\rangle-|110010\rangle \\&& -|001010\rangle +|101010\rangle -|011010\rangle
 -|111010\rangle-|000110\rangle   \\&&+|100110\rangle+|010110\rangle +|110110\rangle -|001110\rangle
 -|101110\rangle   \\&&+|011110\rangle -|111110\rangle+|000001\rangle+|100001\rangle -|010001\rangle  \\&&
 +|110001\rangle +|001001\rangle -|101001\rangle -|011001\rangle -|111001\rangle \\&&-|000101\rangle
 +|100101\rangle -|010101\rangle -|110101\rangle +|001101\rangle \\&& +|101101\rangle +|011101\rangle
-|111101\rangle-|000011\rangle -|100011\rangle   \\&&+|010011\rangle -|110011\rangle +|001011\rangle -|101011\rangle
-|011011\rangle\\&& -|111011\rangle-|000111\rangle +|100111\rangle -|010111\rangle -|110111\rangle   \\&&-|001111\rangle
-|101111\rangle -|011111\rangle +|111111\rangle \\&&
\end{eqnarray*}

The following Table III shows  lower bounds on  $k_d(n)$ via our Corollary \ref{cor:1} and Theorem \ref{thm:2}. Some of entries in Table III are obtained via Corollary \ref{cor:1} from those best-known self-dual codes in the online table \cite{GO}, while some others are obtained from Theorem \ref{thm:2} through algebraic geometry codes and computer search. In particular, all entries for $p=7$ are obtained from algebraic geometry codes. Note that in Table III some entries on $k_p(n)$ for odd $n$ are computed from Corollary \ref{cor:1}(ii).

\begin{center}
Table III\\ \medskip

\begin{tabular}{||c|c|c|c|c|c|c|c|c|c|c|c|c|c|c|c|c|c||} \hline
$n$     &$2$ &$3$ & $4$ & $5$ & $6$ & $7$ &$8$ &$9$ & $10$ & $11$ & $12$ & $13$  &$14$ &$15$ & $16$ & $17$  \\ \hline
$k_2(n)$ &$1$ &$1$ & $1$ & $1$ & $2$ & $2$ &$3$ &$2$ & $2$ & $2$ & $3$ & $2$  &$3$ &$3$ & $3$ & $3$ \\ \hline
$k_3(n)$ &$1$ &$1$ & $2$ & $2$ & $2$ & $2$ &$2$ &$2$ & $3$ & $4$ & $5$ & $3$  &$3$ &$4$ & $5$ & $4$  \\ \hline
$k_5(n)$ &$1$ &$1$ & $2$ & $2$ & $3$ & $2$ &$3$ &$3$ & $4$ & $4$ & $5$ & $4$  &$5$ &$5$ & $6$ & 5 \\ \hline
$k_7(n)$ &$1$ &$1$ & $2$ & $2$ & $3$ & $3$ &$4$ &$3$ & $4$ & $4$ & $5$ & $5$  &$5$ &$5$ & $6$ & $5$   \\ \hline\hline

$n$  & $18$ & $19$   &$20$ &$21$ & $22$ & $23$ & $24$&&&&&&&&&  \\ \hline
$k_2(n)$& $3$ & $3$   &$3$ &$4$ & $5$ & $6$ & $7$ &&&&&&&&& \\ \hline
$k_3(n)$& $4$ & $4$   &$5$ &$5$ & $6$ & $7$ & $8$ &&&&&&& && \\ \hline
$k_5(n)$& $6$   & $6$ &$7$ &$6$ & $7$ & $7$ & $8$ &&&&&&&&&  \\ \hline
$k_7(n)$& $6$ & $7$ &$8$ &$6$ & $7$ & $7$ & $8$ &&&&&&&&&  \\ \hline
\end{tabular}
\end{center}

\begin{rmk}\label{rmk:3} {\rm Comparing Tables I with III, we find that Theorem \ref{thm:1} usually gives the same or better results than Theorem \ref{thm:2}. The only exceptional cases are $k_2(24)$ and $k_7(8)$. This is due to the extreme example of binary Golay $[24,12,8]$-self-dual code and MDS code $\ZZ_7$.
}\end{rmk}

\subsection{Large level $p$}
The main purpose of this subsection is to prove that, for any given even $n$, we have $k_p(n)=n/2$ for sufficiently large prime $p$.
\begin{thm}\label{thm:3}
For any given even integer $n\ge2$, if an odd prime $p$ satisfies $p\ge {n\choose n/2}+1$, then $k_p(n)=n/2$.
\end{thm}
\begin{proof} {\rm By Theorem \ref{thm:1} and Lemma \ref{lm:2}, it is sufficient to show that $N_p(n,n/2)>0 $ under the condition of our theorem. Indeed
  \begin{eqnarray*}N_p(n,n/2)&=&p^{{n\choose n/2}}\left(1-{n\choose n/2}\left(1-\prod_{i=1}^{n/2}\left(1-\frac1{p^i}\right)\right)\right)\\&\ge& p^{{n\choose n/2}}\left(1-{n\choose n/2}\sum_{i=1}^{n/2}\frac1{p^i}\right)\\&>&  p^{{n\choose n/2}}\left(1-\frac{{n\choose n/2}}{p-1}\right). \end{eqnarray*}
 If $p\ge {n\choose n/2}+1$, then $1-\frac{{n\choose n/2}}{p-1}\ge 0$ and hence $N_p(n,n/2)>0$. The desired result follows.
}\end{proof}

\section{The case where $n$ is large}

For a fixed $d\ge 2$, to see how $k_d(n)$ varies as $n$ tends to infinity, we define the following asymptotic quantity.

\begin{defn} {\rm
For a given integer $d\ge 2$, define the asymptotic quantity \[\lambda_d=\limsup_{n\rightarrow\infty}\frac{k_d(n)}{n}\].
}\end{defn}
Obviously, $\lambda_d\leq 1/2$.
 Again, we will study some lower bounds on $\lambda_d$ by constructing $k$-uniform $n$-qudit quantum states of level $d$ via our results in Section 3.  In addition, we give existence bounds and constructive bounds on $\lambda_d$ separately. As one can expect, constructive bounds are usually worse than existence bounds.

\subsection{Existence bounds on $\Gl_p$}
We first provide an existence bound via Lemma  \ref{lm:2}. For any integer $d\ge 2$, the $d$-ary entropy function is defined by
\[H_d(x):=x\log_d(d-1)-x\log_dx-(1-x)\log_d(1-x).\]
By Stirling's formula, we have
\[\lim_{n\rightarrow\infty}\frac{\log_2{n\choose k}}{n}=H_2(\Gl)\qquad \mbox{if $\frac{k}{n}\rightarrow\Gl$}.\]
\begin{thm}\label{thm:4}
Let $\Gl$ be a root of the equation $H_2(x)=(1-2x)\log_2p$. Then $\Gl_p\ge \Gl$.
\end{thm}
\begin{proof}{\rm Choose a very small $\Ge\in(0,\Gl)$. Put $k=\lfloor (\lambda-\Ge) n\rfloor$.  By Theorem \ref{thm:1} and Lemma \ref{lm:2}, if we can show  that $N_p(n,k)>0 $, there exists a $k$-uniform $n$-qudit quantum state of level $d$.
Note that
 \begin{eqnarray*}
 N_p(n,k)&=&p^{{n\choose n/2}}\left(1-{n\choose k}\left(1-\prod_{i=1}^k\left(1-\frac1{p^{n-k-i}}\right)\right)\right)\\
 &\ge& p^{{n\choose n/2}}\left(1-{n\choose k}p^{-n+2k}\sum_{i=1}^k\frac1{p^i}\right)\\ &>&  p^{{n\choose n/2}}\left(1-\frac{{n\choose k}p^{-n+2k}}{p-1}\right). \end{eqnarray*}
 If \begin{equation}\label{eq:2}
 p\ge {n\choose k}p^{-n+2k}+1,\end{equation}
  then  $N_p(n,k)>0$.  Since
  \begin{equation}\label{eq:2}
(1-2(\Gl-\Ge))\log_p-H_2(\Gl-\Ge)>0, \end{equation}
the equation \eqref{eq:2} holds for sufficiently large $n$. This implies that $\Gl_p\ge \Gl-\Ge$ for any small $\Ge$. Letting $\Ge$ tend to $0$ gives the desired result.
}\end{proof}

Based on Theorem \ref{thm:4}, we provide a table for lower bounds on $\Gl_p$ for small $p$ below.

\begin{center}
Table IV\\ \medskip
\begin{tabular}{||c|c|c|c|c|c|c|c||} \hline
$p$ &$2$ &$3$ & $5$ & $7$ & $11$ & $13$ &$17$ \\ \hline
$\Gl_p$ &$0.1705$ &$0.2461$ & $0.3081$ & $0.3360$ & $0.3634$ & $0.3714$ &$0.3821$  \\ \hline

\end{tabular}
\end{center}

Next let us derive a lower bound on $\Gl_p$ from self-dual codes via Corollary \ref{cor:1}.
\begin{thm}\label{thm:5}  One has
\[\lambda_p \ge H_p^{-1}\left(\frac{1}{2}\right),\] where  $H^{-1}_p(y)$ is the inverse function of $H_p(x)$.
\end{thm}
\begin{proof} {\rm By \cite{Mac}, there exists a family of $p$-ary self-dual code achieving the Gilbert-Vrashamov bound, i.e., there exists a family of $p$-ary $[n,n/2,\ge k+1]$-self-dual code such that
$\lim_{n\rightarrow\infty}\frac{k}{n}\rightarrow \delta$, where $H_p(\delta)=\frac{1}{2}$,
i.e., $\delta=H_p^{-1}(\frac{1}{2})$. It follows immediately that  $\lambda_p \ge \delta=H_p^{-1}(\frac{1}{2})$.
}\end{proof}
Based on Theorem \ref{thm:5}, we provide a table for lower bounds on $\Gl_p$ for small $p$ below.

\begin{center}
Table V\\ \medskip
\begin{tabular}{||c|c|c|c|c|c|c|c||} \hline
$p$ &$2$ &$3$ & $5$ & $7$ & $11$ & $13$ &$17$  \\ \hline
$\Gl_p$ &$0.110$ &$0.159$ & $0.210$ & $0.237$ & $0.268$ & $0.278$ &$0.293$   \\ \hline
\end{tabular}
\end{center}
\begin{rmk}{\rm Once again, the asymptotic result also shows that our matrix constriction given in Theorem \ref{thm:1} is in general better than the one from self-dual codes given in Theorem \ref{thm:2}.
}\end{rmk}

\subsection{Constructive bounds on $\Gl_p$}

\begin{defn} {\rm Let $p$ be a prime and let $r\ge 2$ be an integer.
A $\ZZ_p$-basis $\{\alpha_1,\cdots,\alpha_r\}$  of the Galois field $\rGF(p^r)$ is called trace-orthogonal if $\Tr(\alpha_i\alpha_j)=0$ for all $1\le i\neq j\le r$,
where $\Tr$ is the trace map from  $\rGF(p^r)$ to $\ZZ_p$.
}\end{defn}
It is well known that there always exists a  trace-orthogonal basis of $\rGF(p^r)$ over $\ZZ_p$ \cite[Chapter 5]{MP}.  Note that if $\{\alpha_1,\cdots,\alpha_r\}$ is a trace-orthogonal basis of $\rGF(p^r)$ over $\ZZ_p$, then $\Tr(\Ga_i^2)\neq 0$ for all $1\le i\le r$. Otherwise, one would have $\Tr(\Ga_i x)=0$ for all $x\in \rGF(p^r)$ which is impossible.

Now we fix a trace-orthogonal basis $\{\alpha_1,\cdots,\alpha_r\}$ of $\rGF(p^r)$ over $\ZZ_p$. Let $a_i$ denote $\Tr(\Ga_i^2)$. Then $a_i\in\ZZ_p\setminus\{0\}$. Thus, every element of $\Gb$ of $\rGF(p^r)$ can be written as a linear combination $\Gb=\sum_{i=1}^rb_i\Ga_i$ with $b_i\in\ZZ_p$. We denote by $\pi(\Ga)$ and $\pi^{\perp}(\Ga)$ the vectors $(b_1,b_2,\dots,b_r)\in\ZZ_p^r$ and $(a_1b_1,a_2b_2,\dots,a_rb_r)\in\ZZ_p^r$, respectively.
 Extend $\pi$ and $\pi^{\perp}$ to the maps from $\rGF(p^r)^n$ to $\ZZ_p^{rn}$ given by
\begin{eqnarray*}\label{eq:4}&&\pi(u_1,u_2,\dots,u_n)=(\pi(u_1),\pi(u_2),\dots,\pi(u_n));\qquad \\ && \pi^{-1}(u_1,u_2,\dots,u_n)=(\pi^{\perp}(u_1),\pi^{\perp}(u_2),\dots,\pi^{\perp}(u_n)).\end{eqnarray*}

\begin{lemma}\label{lm:3} Let $C$ be a linear code of length $n$ over  $\rGF(p^r)$.   If $C^{\perp}$ is the dual code of $C$, then $\pi^{\perp}(C^{\perp})$ is the dual code of $\pi(C)\in\ZZ_p^{rn}$.
\end{lemma}
\begin{proof}{\rm Let $\bu=(u_1,u_2,\dots,u_n)\in C$ and $\bv=(v_1,v_2,\dots,v_n)\in C^{\perp}$. Then we have
\begin{eqnarray*}0&=&\Tr\left(\sum_{i=1}^nu_iv_i\right)=\sum_{i=1}^n\sum_{j=1}^r\sum_{\ell=1}^ru_{ij}v_{i\ell}\Tr(\Ga_j\Ga_{\ell})\\&=&\sum_{i=1}^n\sum_{j=1}^ru_{ij}v_{ij}a_j
=\pi(\bu)\cdot\pi^{\perp}(\bv).\end{eqnarray*}
This means that $\pi(C)$ and $\pi^{\perp}(C^{\perp})$ are orthogonal. Furthermore,  it is easy to see that the sum of their dimensions over $\ZZ_p$ is $rn$. This completes the proof.
}\end{proof}

It is a well-known result from algebraic geometry codes that, for any prime power $q$, there exists a  family of $q^2$-ary $[n,n/2,\ge k+1]$ codes $\{C\}$ such that $C^{\perp}$ also have the same parameters $[n,n/2,\ge k+1]$ and $\lim_{n\rightarrow\infty}\frac kn=\frac12-\frac{1}{q-1}$ (see \cite{Stich}). Furthermore, this family can be constructed in polynomial times.

\begin{thm}\label{thm:6} For any $p\ge 5$, one has a constructive lower bound on $\Gl_p$ given by
\[\lambda_p\ge\frac14-\frac1{2(p-1)}.\]
\end{thm}
\begin{proof}{\rm
Consider a family of $p^2$-ary $[n,n/2,\ge k+1]$ codes $\{C\}$ such that $C^{\perp}$ also have the same parameters $[n,n/2,\ge k+1]$ and  $\lim_{n\rightarrow\infty}\frac kn=\frac12-\frac{1}{p-1}$. Consider a trace-orthogonal basis of  $\rGF(p^2)$ over $\ZZ_p$ and associated maps $\pi$ and $\pi^{\perp}$ defined in \eqref{eq:4}. Then both $\pi(C)$ and $\pi^{\perp}(C^{\perp})$ are $p$-ary $[2n,n,\ge k+1]$-linear code. By Theorem \ref{thm:2}, we have a $k$-uniform $rn$-qudit quantum state of level $p$. This gives $\Gl_p\ge \lim_{n\rightarrow\infty}\frac{k}{2n}=\frac14-\frac1{2(p-1)}$. This completes the proof.
}\end{proof}
When $p$ is small, the bound given in Theorem \ref{thm:6} can be further improved by considering algebraic geometry codes over larger extension $\rGF(p^{2t})$ for $t\ge 2$.

\begin{thm} \label{thm:7}
For any $t\ge 2$, one has a  constructive lower bound on $\Gl_p$ given by
\[\lambda_p\ge \frac{1}{2t} \left(\frac12-\frac{1}{p^t-1}\right).\]
\end{thm}
\begin{proof}{\rm  The proof is almost identical to the one of Theorem \ref{thm:6} except we consider algebraic geometry codes over larger extension.

Consider a family of $p^{2t}$-ary $[n,n/2,\ge k+1]$ codes $\{C\}$ such that $C^{\perp}$ also have the same parameters $[n,n/2,\ge k+1]$ and  $\lim_{n\rightarrow\infty}\frac kn=\frac12-\frac{1}{p^t-1}$. Consider a trace-orthogonal basis of  $\rGF(p^{2t})$ over $\ZZ_p$ and associated maps $\pi$ and $\pi^{\perp}$ defined in \eqref{eq:4}. Then both $\pi(C)$ and $\pi^{\perp}(C^{\perp})$ are $p$-ary $[2tn,tn,\ge k+1]$-linear code. By Theorem \ref{thm:2}, we have a $k$-uniform $2tn$-qudit quantum state of level $p$. This gives $\Gl_p\ge \lim_{n\rightarrow\infty}\frac{k}{2tn}=\frac1{2t}\left(\frac12-\frac1{p^t-1}\right)$. This completes the proof.
}\end{proof}

Finally, we provide  a table for  constructive lower bounds on $\Gl_p$ for primes $p=2,3,5,\dots, 23$. Note that the value $t$  to obtain the optimal lower bound in Theorem \ref{thm:2}  may vary as $p$ varies.

\begin{center}
Table V\\ \medskip
\begin{tabular}{||c|c|c|c|c|c|c|c|||} \hline
$p$ &$2$ &$3$ & $5$ & $7$ & $11$ & $13$ &$17$   \\ \hline
$\Gl_p$ &$0.060$ &$0.094$ & $0.125$ & $0.167$ & $0.2$ & $0.208$ &$0.219$  \\ \hline\hline
$t$ &$3$ &$2$ & $1$ & $1$ & $1$ & $1$ &$1$  \\ \hline
\end{tabular}
\end{center}

\begin{acknowledgments}
Keqing Feng is supported by NSFC No.11471178,11571007 and the Tsinghua National Lab. on Information Science and Technology.
 Lingfei Jin is supported in part by  Shanghai Sailing Program  under the grant 15YF1401200 and by National Natural Science Foundation of China under Grant 11501117.
Chaoping Xing and Chen Yuan
are supported by the Singapore Ministry of Education Tier 1 under Grant
RG20/13.

\end{acknowledgments}

\appendix

\section{ An example for k-uniform from symmetric matrix}
In this appendix, we provide one more  matrix that gives a 3-uniform 8-qudit quatum states of level 2. Thus, the quantum states can be explicitly written down as long as the corresponding matrix is provided by  Theorem \ref{thm:1}.

\[
\left(
  \begin{array}{cccccccc}
0& 1& 1& 0& 0& 1& 1& 0\\
1& 0& 1& 0& 0& 1& 0& 1\\
1& 1& 0& 1& 1& 0& 1& 1\\
0& 0& 1& 0& 1& 0& 0& 1\\
0& 0& 1& 1& 0& 1& 1& 1\\
1& 1& 0& 0& 1& 0& 0& 1\\
1& 0& 1& 0& 1& 0& 0& 0\\
0& 1& 1& 1& 1& 1& 0& 0\\
  \end{array}
\right)\]

\begin{eqnarray*}
&&|\Phi\rangle=|00000000\rangle +|10000000\rangle +|01000000\rangle -|11000000\rangle \\&&+|00100000\rangle
-|10100000\rangle-|01100000\rangle -|11100000\rangle  \\&&+|00010000\rangle +|10010000\rangle +|01010000\rangle
-|11010000\rangle  \\&&-|00110000\rangle +|10110000\rangle +|01110000\rangle +|11110000\rangle \\&&+|00001000\rangle +|10001000\rangle +|01001000\rangle -|11001000\rangle   \\&&
-|00101000\rangle +|10101000\rangle +|01101000\rangle +|11101000\rangle \\&&-|00011000\rangle -|10011000\rangle -|01011000\rangle +|11011000\rangle  \\&&-|00111000\rangle +|10111000\rangle +|01111000\rangle +|11111000\rangle \\&&+|00000100\rangle -|10000100\rangle -|01000100\rangle -|11000100\rangle \\&&+|00100100\rangle +|10100100\rangle +|01100100\rangle -|11100100\rangle \\&&+|00010100\rangle -|10010100\rangle -|01010100\rangle -|11010100\rangle \\&&-|00110100\rangle -|10110100\rangle -|01110100\rangle +|11110100\rangle \\&&-|00001100\rangle +|10001100\rangle +|01001100\rangle +|11001100\rangle \\&&+|00101100\rangle +|10101100\rangle +|01101100\rangle -|11101100\rangle \\&&+|00011100\rangle -|10011100\rangle -|01011100\rangle -|11011100\rangle \\&&+|00111100\rangle +|10111100\rangle +|01111100\rangle -|11111100\rangle \\&&+|00000010\rangle -|10000010\rangle +|01000010\rangle +|11000010\rangle \\&&-|00100010\rangle -|10100010\rangle +|01100010\rangle -|11100010\rangle \\&&+|00010010\rangle -|10010010\rangle +|01010010\rangle +|11010010\rangle \\&&+|00110010\rangle +|10110010\rangle -|01110010\rangle +|11110010\rangle \\&&-|00001010\rangle +|10001010\rangle -|01001010\rangle -|11001010\rangle \\&&-|00101010\rangle -|10101010\rangle +|01101010\rangle -|11101010\rangle \\&&+|00011010\rangle -|10011010\rangle +|01011010\rangle +|11011010\rangle \\&&-|00111010\rangle -|10111010\rangle +|01111010\rangle -|11111010\rangle \\&&+|00000110\rangle +|10000110\rangle -|01000110\rangle +|11000110\rangle \\&&-|00100110\rangle +|10100110\rangle -|01100110\rangle -|11100110\rangle \\&&+|00010110\rangle +|10010110\rangle -|01010110\rangle +|11010110\rangle \\&&+|00110110\rangle -|10110110\rangle +|01110110\rangle +|11110110\rangle \\&&+|00001110\rangle +|10001110\rangle -|01001110\rangle +|11001110\rangle \\&&+|00101110\rangle -|10101110\rangle +|01101110\rangle +|11101110\rangle \\&&-|00011110\rangle -|10011110\rangle +|01011110\rangle -|11011110\rangle \end{eqnarray*}
\begin{eqnarray*}&&+|00111110\rangle -|10111110\rangle +|01111110\rangle +|11111110\rangle \\&&+|00000001\rangle +|10000001\rangle -|01000001\rangle +|11000001\rangle \\&&-|00100001\rangle +|10100001\rangle -|01100001\rangle -|11100001\rangle \\&&-|00010001\rangle -|10010001\rangle +|01010001\rangle -|11010001\rangle
\\&&-|00110001\rangle +|10110001\rangle -|01110001\rangle -|11110001\rangle \\&&-|00001001\rangle -|10001001\rangle +|01001001\rangle -|11001001\rangle \\&&-|00101001\rangle +|10101001\rangle -|01101001\rangle -|11101001\rangle \\&&-|00011001\rangle -|10011001\rangle +|01011001\rangle -|11011001\rangle \\&&+|00111001\rangle -|10111001\rangle +|01111001\rangle +|11111001\rangle \\&&-|00000101\rangle +|10000101\rangle -|01000101\rangle -|11000101\rangle \\&&+|00100101\rangle +|10100101\rangle -|01100101\rangle +|11100101\rangle \\&&+|00010101\rangle -|10010101\rangle +|01010101\rangle +|11010101\rangle \\&&+|00110101\rangle +|10110101\rangle -|01110101\rangle +|11110101\rangle \\&&-|00001101\rangle +|10001101\rangle -|01001101\rangle -|11001101\rangle \\&&-|00101101\rangle -|10101101\rangle +|01101101\rangle -|11101101\rangle \\&&-|00011101\rangle +|10011101\rangle -|01011101\rangle -|11011101\rangle \\&&+|00111101\rangle +|10111101\rangle -|01111101\rangle +|11111101\rangle \\&&+|00000011\rangle -|10000011\rangle -|01000011\rangle -|11000011\rangle \\&&+|00100011\rangle +|10100011\rangle +|01100011\rangle -|11100011\rangle \\&&-|00010011\rangle +|10010011\rangle +|01010011\rangle +|11010011\rangle \\&&+|00110011\rangle +|10110011\rangle +|01110011\rangle -|11110011\rangle \\&&+|00001011\rangle -|10001011\rangle -|01001011\rangle -|11001011\rangle
\\&&-|00101011\rangle -|10101011\rangle -|01101011\rangle +|11101011\rangle \\&&+|00011011\rangle -|10011011\rangle -|01011011\rangle -|11011011\rangle \\&&+|00111011\rangle +|10111011\rangle +|01111011\rangle -|11111011\rangle \\&&-|00000111\rangle -|10000111\rangle -|01000111\rangle +|11000111\rangle \\&&-|00100111\rangle +|10100111\rangle +|01100111\rangle +|11100111\rangle \\&&+|00010111\rangle +|10010111\rangle +|01010111\rangle -|11010111\rangle \\&&-|00110111\rangle +|10110111\rangle +|01110111\rangle +|11110111\rangle \\&&+|00001111\rangle +|10001111\rangle +|01001111\rangle -|11001111\rangle \\&&-|00101111\rangle +|10101111\rangle +|01101111\rangle +|11101111\rangle \\&&+|00011111\rangle +|10011111\rangle +|01011111\rangle -|11011111\rangle \\&&+|00111111\rangle -|10111111\rangle -|01111111\rangle -|11111111\rangle \\&&
\end{eqnarray*}

\bibliography{apssamp}

\end{document}